\documentclass[11pt]{article}
\usepackage{graphicx}
\usepackage{tabularx}
\usepackage{url}
\usepackage{amsthm}
\usepackage{amsmath}
\usepackage{amsfonts}
\usepackage[labelfont=bf]{caption}

\theoremstyle{plain}
\newtheorem{theorem}{Theorem}
\newtheorem{lemma}{Lemma}
\newtheorem{proposition}{Proposition}

\newtheorem{corollary}{Corollary}
\theoremstyle{definition}
\newtheorem{definition}{Definition}

\theoremstyle{remark}

\date{}

\begin{document}

\title{Programs Versus Finite Tree-Programs}

\author{Mikhail Moshkov \\
Computer, Electrical and Mathematical Sciences \& Engineering Division \\
King Abdullah University of Science and Technology (KAUST) \\
Thuwal 23955-6900, Saudi Arabia\\ mikhail.moshkov@kaust.edu.sa
}

\maketitle

\begin{abstract}
In this paper, we study classes of structures and individual
structures for which programs implementing functions defined everywhere are
equivalent to finite tree-programs. The programs under
consideration may have cycles and at most countably many nodes. We start
with programs in which arbitrary terms of a given signature may be used in
function nodes and arbitrary formulas of this signature may be used in
predicate nodes. We then extend our results to programs that are close in
nature to computation trees: if such a program is a finite tree-program, then it is
an ordinary computation tree.
\end{abstract}


\section{Introduction \label{S6.1}}

Finite tree-programs such as decision trees  \cite{BreimanFOS84,Dobkin76,Moshkov05,Yao97} and computations trees \cite{Ben-Or83,Gabrielov17,Grigoriev96,Moshkov22a} are studied in different areas of computer science. Decision trees can be used as classifiers and predictors in data analysis. Decision and computation trees  can be used as algorithms for solving  problems of combinatorial optimization, computation geometry, etc.

Finite tree-programs are a very simple type of programs that implement functions defined everywhere. However, there are situations for which all programs from sufficiently broad classes that implement functions defined everywhere are equivalent to finite tree-programs. The paper is devoted to the study of such situations.

In this paper, a program is a pair $(S,U)$, where $S$ is a scheme of program of a signature $\sigma$ and $U$ is a structure of the signature $\sigma$. Signature $\sigma$ is a finite or countable set of predicate and function symbols with their arity, and constant symbols. Structure $U$ is a pair $(A,I)$, where $A$ is a nonempty set called the universe of $U$ and $I$ is an interpretation function mapping the symbols of $\sigma$ to predicates, functions, and constants in $A$. The schemes of programs under consideration may have cycles and finite or countable sets of nodes. Of particular interest to us are program schemes that are finite trees. The programs corresponding to them are called  finite tree-programs.

Let $K$ be a nonempty class of structures of the signature $\sigma$. A scheme of program $S$ is called total relative to $K$ if, for any structure $U \in K$, the function implemented by the program $(S,U)$ is everywhere defined. The class $K$  is called program-saturated if any scheme of program of the signature $\sigma$ that is total relative to $K$ is equivalent to a scheme of program that is a finite tree.

We begin our study with the general case, where schemes of programs can use arbitrary terms of a given signature $\sigma$ in function nodes and arbitrary formulas
of that signature in predicate nodes. These schemes of programs are discussed in Section \ref{S6.2}.

Section \ref{S6.3} studies program-saturated classes of structures. A
necessary and sufficient condition for a class of structures to be program-saturated is its compactness: any finitely satisfiable in $K$ set of formulas  with free variables from a finite set  is satisfiable in $K$. In particular, any axiomatizable class of structures has the property of compactness and is therefore a program-saturated class.

Section \ref{S6.4} is devoted to the study of individual structures, each of
which forms a class that is program-saturated. Such structures
include, in particular, all models of cardinality $\alpha $ of  $%
\alpha $-categorical theories. An example of such a model is the field of
complex numbers.

In Section \ref{S6.5}, we study the possibility of  elementary extension of the
structure to a structure that is program-saturated. We show that this is always
possible and find the minimum cardinality of such an extension.

Section \ref{S6.6} is devoted to the transfer of the obtained results to
programs that are essentially close to the computation trees: if the scheme of such a program is a finite tree, then the program is an ordinary
computation tree. Some results in this direction were published earlier
without proofs in \cite{Moshkov87}.

\section{Schemes of Programs and Programs \label{S6.2}}

We begin our study with the general case, where schemes of programs can use arbitrary
terms of a given signature in function nodes and arbitrary formulas of that signature in
predicate nodes.

Let $\omega =\{0,1,2,\ldots \}$ and $X=\{x_{i}:i\in \omega \}$ be the set of
variables. Let $\sigma $ be a finite or countable \emph{signature}: a set of
predicate and function symbols with their arity, and constant
symbols.  The concept of a formula of signature $\sigma $ is defined in a
standard way. First, the concept of a \emph{term} is defined, then the
concept of an \emph{atomic formula} using, in particular, the equality
symbol $=$, and finally the concept of a \emph{formula} using additionally
the logical symbols $\wedge $, $\lnot $ and $\forall $ (see definitions on
pp. 22 and 23 \cite{ChangK92}). For $n\in \omega \setminus \{0\}$, we denote
$X_{n}=\{x_{0},\ldots ,x_{n-1}\}$ and $F_{n}(\sigma )$ the set of formulas
of the signature $\sigma $ with free variables from the set $X_{n}$. Since $%
\sigma $ is finite or countable, for any $n\in \omega \setminus \{0\}$, the
set $F_{n}(\sigma )$ is countable.

\begin{definition}
A \emph{scheme of program} (\emph{scheme} in short) of the signature $\sigma $ is a pair $S=(n,G)$, where $n\in
\omega \setminus \{0\}$ and $G$ is a nonempty directed graph with finite or
countable set of nodes. The nodes of the graph $G$ are divided into three
types: function, predicate, and terminal. A \emph{function} node is labeled
with an expression of the form $x_{i}\Leftarrow t$, where $t$ is a term of
the signature $\sigma $. Only one edge leaves this node and this edge is not labeled. A \emph{predicate}
node is labeled with a formula of the signature $\sigma $. Two edges leave
this node. One edge is labeled with the number $1$ and another one is
labeled with the number $0$. A \emph{terminal} node is labeled with a number
from $\omega $. This node has no leaving edges. In addition, some node of
the graph is selected as the \emph{initial} one and marked with the $\ast $
sign. The set $X_{n}$ will be called the \emph{set of input variables} of the
scheme $S$. Usually, we will not distinguish between a scheme and its graph.
\end{definition}
\begin{definition}
Let $S=(n,G)$ be a scheme of the signature $\sigma $. A \emph{complete path}
of the scheme $S$ is a directed path that starts at the initial node of the
scheme $S$ and is either infinite or ends at a terminal node of the scheme
$S$.
\end{definition}

Let $\tau =w_{1},d_{1},w_{2},d_{2},\ldots $ be a complete path of the scheme
$S$. For $i=1,2,\ldots $, we correspond to the node $w_{i}$ of the path $%
\tau $ a sequence $M_{i}=t_{i0},t_{i1},\ldots $ of terms of the signature $%
\sigma $ with variables from $X_{n}$. Set $M_{1}=x_{0},x_{1},\ldots
,x_{n-2},x_{n-1},x_{n-1},x_{n-1},\ldots $. Let the sequences $M_{1},\ldots
,M_{i}$ already be defined. We now define the sequence $M_{i+1}$ associated
with the node $w_{i+1}$. If $w_{i}$ is a predicate node, then $M_{i+1}=M_{i}$%
. Let $w_{i}$ be a functional node labeled with the expression $%
x_{j}\Leftarrow t(x_{l_{1}},\ldots ,x_{l_{h}})$. Then $M_{i+1}=t_{i0},\ldots
,t_{ij-1},t(t_{il_{1}},\ldots ,t_{il_{h}}),t_{ij+1},\ldots $.

We correspond to each predicate node of the path $\tau $ a formula from the
set $F_{n}(\sigma )$. Let $w_{i}$ be a predicate node that is labeled with a
formula $\varphi (x_{l_{1}},\ldots ,x_{l_{q}})$ and let the edge $d_{i}$
of the path $\tau $ leaving the node $w_{i}$ be labeled with the number $%
c$. Then we correspond to $w_{i}$ the formula $\varphi (t_{il_{1}},\ldots
,t_{il_{q}})$ if $c=1$ and the formula $\lnot \varphi (t_{il_{1}},\ldots
,t_{il_{q}})$ if $c=0$.

Denote by $\Pi (\tau )$ the set of formulas corresponded to the predicate
nodes of the complete path $\tau $. If $\tau $ is a finite path, then denote
by $t_{\tau }$ the number attached to the terminal node of this path.

A \emph{structure} of the signature $\sigma $ is a pair $U=(A,I)$, where $A$
is a nonempty set called the \emph{universe} of the structure and\ $I$ is an
\emph{interpretation function} mapping the symbols of $\sigma $ to
appropriate predicates, functions and constants in $A$.  The cardinal $|A|$
is called the \emph{cardinality} of the structure $U$. Let $\varphi
\in F_{n}(\sigma )$ and $\bar{a}\in A^{n}$. The expression $U\models \varphi
(\bar{a})$ means that the formula $\varphi $ is \emph{true} in the structure
$U$ for the tuple $\bar{a}$. We will say that the formula $\varphi $ is
\emph{satisfiable} in the structure $U$ if $U\models \varphi (\bar{a})$ for
some $\bar{a}\in A^{n}$. We will say that a set of formulas $\Phi \subseteq
F_{n}(\sigma )$ is \emph{satisfiable} in the structure $U$ if there is a
tuple $\bar{a}\in A^{n}$ such that $U\models \varphi (\bar{a})$ for any
formula $\varphi \in \Phi $. Let $\varphi $ be a \emph{sentence} of the
signature $\sigma $, i.e., a formula of the signature $\sigma $ without free
variables. The expression $U\models \varphi $ means that the sentence $%
\varphi $ is \emph{true} in the structure $U$.

\begin{definition}
Let $S=(n,G)$ be a scheme of the signature $\sigma .$ The pair $\Gamma
=(S,U) $ will be called a \emph{program} over the structure $U$ with the set
of input variables $X_{n}$.
\end{definition}

The scheme $S$ will be called the \emph{scheme}
of the program $\Gamma $. A complete path $\tau $ of the scheme $S$ will be
called \emph{satisfiable} in $U$ on the tuple $\bar{a}\in A^{n}$ if,
for any formula $\varphi \in \Pi (\tau )$, $U\models \varphi (\bar{a})$. One
can show that there exists exactly one complete path of $S$ that is
satisfiable in $U$ on the tuple $\bar{a}$.

We correspond to the program $\Gamma $ possibly partial function $\pi
_{\Gamma }:A^{n}\rightarrow \omega $. Let $\bar{a}\in A^{n}$ and $\tau $ be
a complete path of $S$ that is satisfiable in $U$ on the tuple $\bar{a}$. If
$\tau $ is a finite path, then $\pi _{\Gamma }(\bar{a})=t_{\tau }$. If $\tau
$ is an infinite path, then the value $\pi _{\Gamma }(\bar{a})$ is
undefined. We will say that the program $\Gamma $ \emph{implements} the
function $\pi _{\Gamma }$.

\begin{definition}
A scheme $(n,G)$ of the signature $\sigma $ will be called a \emph{tree-scheme} of the
signature $\sigma $ if $G$ is a tree with the root that coincides with the initial node of $G$. A tree-scheme $(n,G)$ of the signature $\sigma $ will be called \emph{finite} if  $G$ is a finite tree.
\end{definition}

\begin{definition}
Let $S$ be a finite tree-scheme of the signature $\sigma$ and $U$ be a structure of the signature $\sigma$. Then the program $(S,U)$ will be called a \emph{finite tree-program} over the structure $U$.
\end{definition}

\section{Program-saturated Classes of Structures \label{S6.3}}

In this section, we study classes of structures that are program-saturated. We prove that a necessary and sufficient condition for a class of structures to be program-saturates is its  compactness. In particular, any axiomatizable class of
structures has the property of compactness.

\begin{definition}
Let $K$ be a nonempty class of structures of the signature $\sigma $. We
will say that a scheme $S$ of the signature $\sigma $ is \emph{total}
relative to $K$ if, for any structure $U$ from $K$, the program $(S,U)$
implements a total (everywhere defined) function.
\end{definition}

Let $S_{1}=(n_{1},G_{1})$\text{ and }$S_{2}=(n_{2},G_{2})$ be schemes of the
signature $\sigma $. We will say that complete paths $\tau
_{1}=w_{1}^{1},d_{1}^{1},w_{2}^{1},d_{2}^{1},\ldots $ and $\tau
^{2}=w_{1}^{2},d_{1}^{2},w_{2}^{2},d_{2}^{2},\ldots $ of schemes $S_{1}$ and
$S_{2}$ are \emph{isomorphic} if they have the same length and, for $%
i=1,2,\ldots $, the nodes $w_{i}^{1}$ and $w_{i}^{2}$ are labeled with the
same expressions or numbers if they are terminal and either the edges $d_{i}^{1}$ and $d_{i}^{2}$ are not
labeled or they are labeled with the same numbers. We will say that the
schemes $S_{1}$ and $S_{2}$ are \emph{strongly equivalent} relative to the
class $K$ if $n_{1}=n_{2}$ and, for any structure $U$ from $K$ and for any
tuple $\bar{a}\in A^{n_{1}}$, where $A$ is the universe of the structure $U$%
, the complete paths of the schemes $S_{1}$ and $S_{2}$ satisfiable in $U$
on the tuple $\bar{a}$ are isomorphic.

\begin{definition}
The class $K$ will be called \emph{%
program-saturated} if any total relative to $K$ scheme of the signature $%
\sigma $ is strongly equivalent relative to $K$ to a finite tree-scheme of the
signature $\sigma $.
\end{definition}

Let $\Phi \subseteq F_{n}(\sigma )$. We will say that the set of formulas $%
\Phi $ is \emph{satisfiable} in the class $K$ if there exists a structure $%
U\in K$ in which this set of formulas is satisfiable. We will say that the
set $\Phi $ is \emph{finitely satisfiable} in the class $K$ if any finite
subset of $\Phi $ is satisfiable in the class $K$.

\begin{definition}
We will say that the
class $K$ has the \emph{property of compactness }if, for any $n\in \omega
\setminus \{0\}$, any finitely satisfiable in $K$ set of formulas $\Phi
\subseteq F_{n}(\sigma )$ is satisfiable in $K$.
\end{definition}

A complete path $\tau $ of the scheme $S$ will be called \emph{satisfiable}
in the class $K$ if the set of formulas $\Pi (\tau )$ is satisfiable in the
class $K$.

Let $S=(n,G_{1})$ be a scheme of the signature $\sigma $. We denote by $%
R(S)=(n,G_{2})$ a tree-scheme of the signature $\sigma $, which has the following
property: there exists one-to-one correspondence between the set of complete
paths of $S$ and the set of complete paths of $R(S)$ for which the
corresponding paths are isomorphic. Denote by $C$ a subgraph of the graph $%
G_{2}$ induced by the set of nodes that belong to the complete paths of the
scheme $R(S)$, which are satisfiable in the class $K$. For each predicate
node of the graph $C$ with one leaving edge, we add to $C$ a new
terminal node labeled with the number $0$ into which we draw the
missing edge. We denote the obtained graph $G_{3}$. Denote $R(S,K)=(n,G_{3})$%
. Evidently, $R(S,K)$ is a tree-scheme of the signature $\sigma $.

\begin{lemma}
\label{L6.1}Let $K$ be a nonempty class of structures of the signature $%
\sigma $, which has the property of compactness, and $S$ be a scheme of
the signature $\sigma $ that is total relative to $K$. Then all satisfiable
in the class $K$ complete paths of the scheme $S$ are finite and the set of
satisfiable in the class $K$ complete paths of the scheme $S$ is finite.
\end{lemma}

\begin{proof}
Evidently, all satisfiable in the class $K$ complete paths of the scheme $S$
are finite. Let us assume that the set of satisfiable in the class $K$
complete paths of the scheme $S$ is infinite. Similar to the proof of K\"{e}%
nig's lemma \cite{Kleene67}, one can show that there is an infinite complete
path $\tau $ of the scheme $R(S)$ such that each node of the path $\tau $
belongs to a satisfiable in the class $K$ complete paths of the scheme $R(S)$%
. Therefore the set of formulas $\Pi (\tau )$ is finitely satisfiable in $K$%
. Taking into account that the class $K$ has the property of compactness, we
obtain that the set of formulas $\Pi (\tau )$ is satisfiable in $K$. Thus,
there is an infinite complete path in $S$ that is satisfiable in the set $K$
but this is impossible.
\end{proof}

\begin{theorem}
\label{T6.1} A nonempty class $K$ of structures of the signature $\sigma $
is program-saturated if and only if it has the property of compactness.
\end{theorem}

\begin{proof}
Let $K$ have the property of compactness and $S$ be a total relative to $K$
scheme of the signature $\sigma $. Using Lemma \ref{L6.1}, we obtain that $%
R(S,K)$ is a finite tree-scheme of the signature $\sigma $. Evidently, the scheme $%
S $ is strongly equivalent relative to the class $K$ to the scheme $R(S,K)$.
Therefore $K$ is a program-saturated class.

Let $K$ have no the property of compactness. Then there exists a number $%
n\in \omega \setminus \{0\}$ and finitely satisfiable in the class $K$ set
of formulas $\Phi \subseteq F_{n}(\sigma )$, which is not satisfiable in the
class $K$. Let $\Phi =\{\varphi _{1},\varphi _{2},\ldots \}$. We denote by $%
S=(n,G)$ a scheme of the signature $\sigma $, where $G$ is the graph
depicted in Fig. \ref{Fig6.1}. Since the set $\Phi $ is not satisfiable in
the class $K$, the scheme $S$ is total relative to $K$, Taking into account
that the set $\Phi $ is finitely satisfiable in the class $K$, one can show
that the set of complete path of $S$ that are satisfiable in the class $K$
is infinite. Therefore there is no a finite tree-scheme of the signature $\sigma $,
which is strongly equivalent relative to the class $K$ to the scheme $S$.
Thus, the class $K$ is not program-saturated.
\end{proof}

\begin{figure}[th]
\centering
\includegraphics[width=0.60\textwidth]{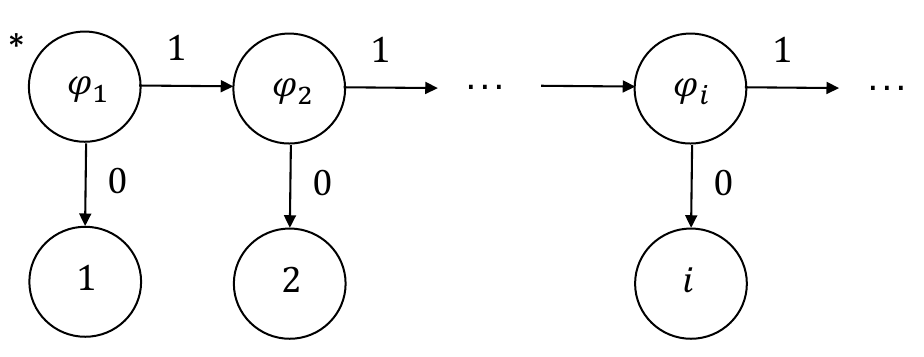}
\caption{Graph $G$}
\label{Fig6.1}
\end{figure}

A \emph{theory} of the signature $\sigma $ is a nonempty set $T$ of
sentences of the signature $\sigma $. A \emph{model} of the theory $T$ is a
structure of the signature $\sigma $ for which all sentences from $T$ are
true. The theory $T$ is called \emph{complete} if, for any sentence $\varphi
$ of the signature $\sigma $, either $\varphi \in T$ or $\lnot \varphi \in
T. $

\begin{definition}
A class $K$ of structures of the signature $\sigma $ is called \emph{%
axiomatizable} if there exists a theory $T$ of the signature $\sigma$ such that
the set of models of $T$ coincides with the class $K$.
\end{definition}

The next lemma follows directly from Proposition 2.2.7 \cite{ChangK92}.

\begin{lemma}
\label{L6.R.1} Let $T$ be a theory of the signature $\sigma $ and $\Phi
\subseteq F_{n}(\sigma )$. Then the following statements are equivalent:

{\rm (a)} $T$ has a model in which $\Phi $ is satisfiable.

{\rm (b)} Every finite subset of $\Phi $ is satisfiable in some model of $T$.
\end{lemma}

\begin{theorem}
\label{T6.2} Any nonempty axiomatizable class $K$ of structures of the
signature $\sigma $ is program-saturated.
\end{theorem}

\begin{proof}
Let $K$ be a nonempty axiomatizable class of structures of the signature $%
\sigma $, $n\in \omega \setminus \{0\}$, $\Phi \subseteq F_{n}(\sigma )$, and
the set of formulas $\Phi $ be finitely satisfiable in the class $K$. Using
Lemma \ref{L6.R.1}, we obtain that the set $\Phi $ is satisfiable in the
class $K$. Therefore the class $K$ has the property of compactness. By
Theorem \ref{T6.1}, the class $K$ is program-saturated.
\end{proof}

We now consider examples of axiomatizable classes of structures (see pp.
38-41 \cite{ChangK92}):

\begin{itemize}
\item Classes of Boo\-lean algebras, atomic Boolean algebras, and atomless
Boo\-lean algebras.

\item Classes of groups, abelian groups, abelian groups with all elements of
order $p$, where $p$ is a prime, torsion-free abelian groups.

\item Classes of commutative rings with unit, fields, fields of
characteristic $p$, where $p$ is a prime, fields of characteristic zero,
algebraically closed fields, real closed fields.
\end{itemize}

\section{Program-saturated Structures \label{S6.4}}

This section is devoted to the study of individual structures, each of which
forms a class that has the property of compactness. Such structures include,
in particular, all models of cardinality $\alpha $ of  $\alpha $%
-categorical theories. An example of such a model is the field of complex
numbers.

\begin{definition}
Let $U$ be a structure of the signature $\sigma $. The structure $U$ will be
called \emph{program-saturated} if the class $\{U\}$ is program-saturated.
\end{definition}

We denote by $\mathrm{Th}(U)$ the \emph{theory} of the structure $U$: the
set of all sentences of the signature $\sigma $ that are true in $U$. It is
clear that the theory $\mathrm{Th}(U)$ is complete. Let $n\in \omega
\setminus \{0\}$ and $\Phi \subseteq F_{n}(\sigma )$.

\begin{definition}
We will say that the
set $\Phi $ is \emph{consistent} with the theory $\mathrm{Th}(U)$ if there
exists a model of the theory $\mathrm{Th}(U)$ in which the set $\Phi $ is
satisfiable.
\end{definition}

\begin{lemma}
\label{L6.2} A set $\Phi \subseteq F_{n}(\sigma )$ is consistent with the
theory $\mathrm{Th}(U)$ if and only if the set $\Phi $ is finitely
satisfiable in the class $\{U\}$.
\end{lemma}

\begin{proof}
Let $\Phi $ be finitely satisfiable in the class $\{U\}$. Using Lemma \ref%
{L6.R.1}, we obtain that $\Phi $ is consistent with the theory $\mathrm{Th}%
(U)$.

Let $\Phi $ be consistent with the theory $\mathrm{Th}(U)$. Let $\varphi
_{1},\ldots ,\varphi _{m}\in \Phi $. Using  the
completeness of the theory $\mathrm{Th}(U)$, we obtain that the sentence $%
\exists x_{0}\cdots \exists x_{n-1}(\varphi _{1}\wedge \cdots \wedge \varphi
_{m})$ belongs to the theory $\mathrm{Th}(U)$. Therefore the set of formulas
$\{\varphi _{1},\ldots ,\varphi _{m}\}$ is satisfiable in $U$. Thus, the set
$\Phi $ is finitely satisfiable in the class $\{U\}$.
\end{proof}

Let $U$ be a structure of the signature $\sigma $ with the universe $A$ and $%
Y\subseteq A$. We denote by $\sigma _{Y}$ the signature obtained from $%
\sigma $ by adding constant symbol  $c_{a}$ for any element $a\in Y$. We denote by $%
U_{Y}$ the expansion of the structure $U$ to the signature $\sigma _{Y}$
such that each new constant symbol $c_{a}$ is interpreted as the element $a$%
.

\begin{definition}
Let $\alpha $ be a cardinal. The structure $U$ is called $\alpha $-\emph{%
saturated} if, for any set $Y\subseteq A$ with cardinality less than $\alpha
$, for any set of formulas $\Phi \subseteq F_{1}(\sigma _{Y})$, which is
consistent with the theory $\mathrm{Th}(U_{Y})$, the set $\Phi $ is
satisfiable in $U_{Y}$. The structure $U$ is called \emph{saturated} if it is
$|A|$-saturated.
\end{definition}

\begin{definition}
Let $\alpha $ be a cardinal. A theory $T$ of the signature $\sigma $ is
called $\alpha $-\emph{categorical} if $T$ has a model of cardinality $%
\alpha $ and every two models of $T$ of the cardinality $\alpha $ are
isomorphic.
\end{definition}

The next lemma follows directly from Proposition 2.3.6 \cite{ChangK92}.

\begin{lemma}
\label{L6.R.2} Let $U$ be an $\omega $-saturated structure of the signature $%
\sigma $ with the universe $A$ and $n \in \omega \setminus \{0\}$. Then, for each finite $Y\subseteq A$, each
set of formulas $\Phi \subseteq F_{n}(\sigma _{Y})$ consistent with the
theory $\mathrm{Th}(U_{Y})$ is satisfiable in $U_{Y}$.
\end{lemma}

Using statement (b) from the proof of Theorem 2.3.13 \cite{ChangK92}, we
obtain the following lemma.

\begin{lemma}
\label{L6.R.3} Let $T$ be a complete $\omega $-categorical theory of the
signature $\sigma $. Then $T$ has a countable $\omega $-saturated model.
\end{lemma}

The next lemma follows directly from Corollary 7.1.8 \cite{ChangK92}.

\begin{lemma}
\label{L6.R.4} Let $\alpha $ be an uncountable cardinal and $T$ be a complete
theory of the signature $\sigma $, which has infinite models. Then $T$ is $%
\alpha $-categorical if and only if every model of $T$ of cardinality $%
\alpha $ is saturated.
\end{lemma}

\begin{theorem}
\label{T6.3} Let $U$ be a structure of the signature $\sigma $ and $\alpha$ be a cardinal.

{\rm (a)} If $U$ is an $\omega $-saturated structure, then $U$ is a
program-saturated structure.

{\rm (b)} If $U$ is a model of the cardinality $\alpha $ of an $\alpha $%
-categorical theory $T$, then $U$ is a program-saturated structure.
\end{theorem}

\begin{proof}
(a) Let $U$ be an $\omega $-saturated structure. Using Lemma \ref{L6.R.2},
we obtain that, for any $n\in \omega \setminus \{0\}$, any consistent with
the theory $\mathrm{Th}(U)$ set of formulas $\Phi \subseteq F_{n}(\sigma )$
is satisfiable in $U$. Using Lemma \ref{L6.2}, we obtain that the class $%
\{U\}$ has the property of compactness. From here and from Theorem \ref{T6.1}
it follows that $U$ is a program-saturated structure.

(b) Let $U$ be a model of the cardinality $\alpha $ of an $\alpha $%
-categorical theory $T$. If $U$ is a finite structure, then evidently, $U$
is a program-saturated structure.

Let $\alpha $ be an infinite cardinal. Evidently, $\mathrm{Th}(U)$ is an $%
\alpha $-categorical complete theory. If $U$ is a countable structure, then
using Lemma \ref{L6.R.3} and the completeness of the theory $\mathrm{Th}(U)$%
, we obtain that $U$ is an $\omega $-saturated structure. If $\alpha $ is an
uncountable cardinal, then using Lemma \ref{L6.R.4} and the completeness of
the theory $\mathrm{Th}(U)$, we obtain that $U$ is saturated and
consequently $\omega $-saturated structure. Using the statement (a) of the
theorem, we obtain that $U$ is a program-saturated structure.
\end{proof}

\begin{corollary}
\label{C6.1} The following structures (see definitions in \cite{ChangK92}),
each of which for some cardinal $\alpha $ is a model of the
cardinality $\alpha $ of an $\alpha $-categorical theory, are
program-saturated:

\begin{itemize}
\item Countable atomless Boolean algebra (see Proposition 1.4.5 \cite%
{ChangK92}).

\item Abelian group with all elements of order $p$, where $p$ is a prime number (see
Proposition 1.4.7 \cite{ChangK92}).

\item Uncountable divisible torsion-free abelian group (see Proposition
1.4.8 \cite{ChangK92}), in particular, additive group of real numbers $(%
\mathbb{R}
;+,0)$.

\item Uncountable algebraically closed field of the characteristic zero or $%
p $, where $p$ is a prime number (see Proposition 1.4.10 \cite{ChangK92}),
in particular, the field of complex numbers $(%
\mathbb{C}
;+,\cdot ,0,1)$.
\end{itemize}
\end{corollary}

\section{Elementary Extensions \label{S6.5}}

In this section, we study the possibility of  elementary extension of the
structure to a structure that is program-saturated. We show that this is always
possible and find the minimum cardinality of such an extension.

\begin{definition}
Let $n\in \omega \setminus \{0\}$ and $F_{n}(\sigma )=\{\varphi _{i}:i\in
\omega \}$. $n$-\emph{Type} of the structure $U$ of the signature $\sigma $
is any finitely satisfiable in $\{U\}$ set of formulas of the form $\{\varphi
_{i}^{\delta _{i}}:i\in \omega \}$, where, for any $i\in \omega $, $\delta
_{i}\in \{0,1\}$ and $\varphi _{i}^{1}=\varphi _{i}$, $\varphi
_{i}^{0}=\lnot \varphi _{i}$.
\end{definition}

\begin{lemma}
\label{L6.3} A structure $U$ of the signature $\sigma $ is program-saturated
if and only if, for any $n\in \omega \setminus \{0\}$, any $n$-type of the
structure $U$ is satisfiable in $U$.
\end{lemma}

\begin{proof}
Let the structure $U$ be program-saturated. Using Theorem \ref{T6.1}, we
obtain that, for any $n\in \omega \setminus \{0\}$, any $n$-type of the
structure $U$ is satisfiable in $U$.

Let, for any $n\in \omega \setminus \{0\}$, any $n$-type of the structure $U$
be satisfiable in $U$. Let $\Phi \subseteq F_{n}(\sigma )$ and the set $\Phi
$ be finitely satisfiable in $\{U\}$. Using Lemma \ref{L6.2}, we obtain that
$\Phi $ is satisfiable in some model $U_{1}$ of the theory $\mathrm{Th}(U)$
on some tuple $\bar{a}\in A_{1}^{n}$, where $A_{1}$ is the universe of the
structure $U_{1}$. Denote $H=\{\varphi _{i}^{\delta _{i}}:i\in \omega \}$,
where, for any $i\in \omega $, $\delta _{i}\in \{0,1\}$ and $U_{1}\models
\varphi _{i}^{\delta _{i}}(\bar{a})$. Using Lemma \ref{L6.2}, we obtain that
$H$ is an $n$-type of the structure $U$. Hence the set $H$ is
satisfiable in $U$. Taking into account that $\Phi \subseteq H$, we obtain
that $\Phi $ is satisfiable in $U$. Thus, the class $\{U\}$ has the property
of compactness. Using Theorem \ref{T6.1}, we obtain that $U$ is a
program-saturated structure.
\end{proof}

Let $U_{1}$ and $U_{2}$ be structures of the signature $\sigma $ with
universes $A_{1}$ and $A_{2}$, respectively.

\begin{definition}
The structure $U_{2}$ is an
\emph{extension} of the structure $U_{1}$ if $A_{1}\subseteq A_{2}$, each
predicate of $U_{1}$ is the restriction of corresponding predicate of $U_{2}$
to $A_{1}$, each function of $U_{1}$ is the restriction of corresponding
function of $U_{2}$ to $A_{1}$, and each constant of $U_{1}$ is the
corresponding constant of $U_{2}$. The notation $U_{1}\subset U_{2}$ means
that $U_{2}$ is an extension of $U_{1}$. If  $U_{1}\subset U_{2}$, we will say that
$U_1$ is a \emph{substructure} of $U_2$.
\end{definition}

\begin{definition}
We will say that $U_{2}$ is an \emph{elementary extension} of $U_{1}$ if $%
U_{2}$ is an extension of $U_{1}$ and, for any formula $\varphi \in
F_{n}(\sigma )$ and any $n$-tuple $\bar{a}\in A_{1}^{n}$, $U_{1}\models
\varphi (\bar{a})$ if and only if $U_{2}\models \varphi (\bar{a})$. The
notation $U_{1}\prec U_{2}$ means that $U_{2}$ is an elementary extension of
$U_{1}$.  If  $U_{1}\prec U_{2}$, we will say that
$U_1$ is an \emph{elementary substructure} of $U_2$.
\end{definition}

\begin{definition}
The structures $U_{1}$ an $U_{2}$ are called \emph{elementary equivalent}
if, for any sentence $\varphi $ of the signature $\sigma $, $U_{1}\models
\varphi $ if and only if $U_{2}\models \varphi $. The notation $U_{1}\equiv
U_{2}$ means that the structures $U_{1}$ an $U_{2}$ are elementary
equivalent.
\end{definition}

Next lemma follows directly from Proposition 3.1.1 \cite{ChangK92}.

\begin{lemma}
\label{L6.R.5} \emph{(a)} If $U_{1}\prec U_{2}$, then $U_{1}\equiv U_{2}$.

{\rm (b)} If $U_{1}\prec U_{3}$, $U_{2}\prec U_{3}$ and $U_{1}\subset
U_{2}$, then $U_{1}\prec U_{2}$.
\end{lemma}

Let $U$ be a structure of the signature $\sigma $ and $n\in \omega \setminus \{0\}$. We denote by $S_{n}U$ the set of all $n$-types of the structure $U$.

\begin{lemma}
\label{L6.5} If a structure $U_{2}$ is an elementary extension of a structure
$U_{1}$, then $S_{n}U_{1}=S_{n}U_{2}$.
\end{lemma}

\begin{proof}
Let $U_{1}\prec U_{2}$. Using statement (a) of Lemma \ref{L6.R.5}, we obtain
that $\mathrm{Th}(U_{1})=\mathrm{Th}(U_{2})$. From here and from Lemma \ref%
{L6.2} it follows that, for any $n\in \omega \setminus \{0\}$, $%
S_{n}U_{1}=S_{n}U_{2}$.
\end{proof}

Next lemma follows directly from Lemma 3.4 \cite{Jech02}.

\begin{lemma}
\label{L6.R.6} If $Y$ is a set of cardinals, then $\sup Y$ is a cardinal.
\end{lemma}

Let $\alpha $ be a cardinal. We denote by $\alpha ^{+}$ the\emph{\ successor}
of the cardinal $\alpha $, i.e., the least cardinal that is greater than $%
\alpha $. Denote $||\sigma ||=\omega \cup |\sigma |$.

Next lemma follows directly from Lemma 5.1.4 \cite{ChangK92}.

\begin{lemma}
\label{L6.R.7} Let $U$ be a structure of the signature $\sigma $ with the
universe $A$ and $\alpha $ be a cardinal such that $||\sigma ||\leq \alpha $
and $\omega \leq |A|\leq 2^{\alpha }$. Then there is an $\alpha ^{+}$%
-saturated elementary extension of $U$ of the cardinality $2^{\alpha }$.
\end{lemma}

Next lemma follows directly from Theorem 3.1.6 \cite{ChangK92}.

\begin{lemma}
\label{L6.R.8} Let $U$ be a structure of the signature $\sigma $ with the
universe $A$, $|A|=\alpha $ and $||\sigma ||\leq \beta \leq \alpha $. Then, for
any set $B\subseteq A$ with cardinality at most $\beta $, the structure $%
U $ has an elementary substructure of cardinality $\beta $, which universe
contains $B$.
\end{lemma}

Next lemma follows directly from Lemma 5.2 \cite{Jech02}.

\begin{lemma}
\label{L6.R.9} 
Let $S$ be a family of sets. Then 
$|\bigcup_{P \in S} P| \le |S| \cdot \sup \{|P|: P \in S\}$.
\end{lemma}

Let $U$ be a structure of the signature $\sigma $ with the universe $A$. For
$n\in \omega \setminus \{0\}$, we denote  $\alpha _{n}(U)=|S_{n}U|$. If the set $A$ is a finite
set, we denote $\alpha (U)=|A|$. If $A$ is an infinite set, then denote $%
\alpha (U)=\sup \{|A|,\alpha _{1}(U),\alpha _{2}(U),\ldots \}$. By Lemma \ref%
{L6.R.6}, $\alpha (U)$ is a cardinal.

\begin{theorem}
\label{T6.4} Let $U_{1}$ be a structure of the signature $\sigma $.

{\rm (a)} The cardinality of any program-saturated structure, which is an
elementary extension of the structure $U_{1}$, is greater than or equal to $%
\alpha (U_{1})$.

{\rm (b)} There exists a program-saturated structure $U_{2}$, which is an
elementary extension of the structure $U_{1}$ and which cardinality is equal
to $\alpha (U_{1})$.
\end{theorem}

\begin{proof}
Let $A_{1}$ be the universe of the structure $U_{1}$.

(a) Let $U_{2}$ be a program-saturated elementary extension of the structure
$U_{1}$ and $A_{2}$ be the universe of the structure $U_{2}$. Evidently, $%
|A_{2}|\geq |A_{1}|$. Therefore if $A_{1}$ is a finite set, then $%
|A_{2}|\geq \alpha (U_{1})$. Let $A_{1}$ be an infinite set. Let $n\in
\omega \setminus \{0\}$. Since $U_{2}$ is a program-saturated structure, by
Lemma \ref{L6.3}, any $n$-type from $S_{n}U_{2}$ is satisfiable in $U_{2}$.
Evidently, different $n$-types from $S_{n}U_{2}$ are satisfiable on
different $n$-tuples from $A_{2}^{n}$. Therefore $|A_{2}^n|\geq |S_{n}U_{2}|$.
Since $A_{1}\subseteq A_{2}$ and $A_{1}$ is an infinite set, $%
|A_{2}^{n}|=|A_{2}|$. Using Lemma \ref{L6.5}, we obtain $|S_{n}U_{2}|=\alpha
_{n}(U_{1})$. Therefore, $|A_{2}|\geq \alpha _{n}(U_{1})$. Thus, $%
|A_{2}|\geq \alpha (U_{1})$.

(b) Let $A_{1}$ be a finite set. Evidently, $U_{1}$ is a program-saturated
structure, $U_{1}\prec U_{1}$ and the cardinality of $U_{1}$ is equal to $%
\alpha (U_{1})$. Therefore in the capacity of $U_{2}$ we can take $U_{1}$.

Let $A_{1}$ be an infinite set. From Lemma \ref{L6.R.7} it follows that
there exists an $\omega $-saturated elementary extension $U_{3}$ of the
structure $U_{1}$. From Theorem \ref{T6.3} it follows that $U_{3}$ is a
program-saturated structure. Using Lemma \ref{L6.3}, we obtain that, for any $n\in
\omega \setminus \{0\}$, any $n$-type $H\in S_{n}U_{3}$ is satisfiable on
some $n$-tuple $\bar{a}_{H}\in A_{3}^{n}$, where $A_{3}$ is the universe of $%
U_{3}$.

For an arbitrary $n\in \omega \setminus \{0\}$, we denote $B_{n}=\{%
\bar{a}_{H}:H\in S_{n}U_{3}\}$ and $C_{n}$ the set of all elements from $%
A_{3}$ belonging to the $n$-tuples from $B_{n}$. If $C_{n}$ is a finite set,
then $|C_{n}|<|A_{1}|\leq \alpha (U_{1})$. If $C_{n}$ is an infinite set,
then as it is not difficult to check, $|C_{n}|=|B_{n}|$. Evidently, $%
|B_{n}|\leq |S_{n}U_{3}|$. Using Lemma \ref{L6.5}, we obtain $%
|S_{n}U_{3}|=\alpha _{n}(U_{1})$. Hence $|C_{n}|\leq \alpha _{n}(U_{1})\leq
\alpha (U_{1})$. Denote $C=\bigcup_{n=1}^{\infty }C_{n}$.
Taking into account that $\alpha
(U_{1})$ is an infinite cardinal and, for any $n\in \omega \setminus \{0\}$,
$|C_{n}|\leq \alpha (U_{1})$, and using Lemma \ref{L6.R.9},
we obtain $|C|\leq $ $\alpha (U_{1})$.

Denote $Y=A_{1}\cup C$, $\beta =|Y|$, and $\alpha =|A_3|$.
Evidently, $Y\subseteq A_{3}$. Therefore $||\sigma||=\omega \le \beta \le \alpha$.
 Using Lemma \ref{L6.R.8}, we obtain
that there exists an elementary substructure $U_{2}$ of the structure $U_{3}$
such that $Y\subseteq A_{2}$ and $|A_{2}|=\beta $, where $A_{2}$ is the
universe of $U_{2}$. Using Lemma \ref{L6.5}, we obtain that, for any $n\in
\omega \setminus \{0\}$, $S_{n}U_{2}=S_{n}U_{3}$. Let $n\in \omega \setminus
\{0\}$. Since $C\subseteq $ $A_{2}$, $U_{2}\prec U_{3}$, and any $n$-type $%
H\in S_{n}U_{3}$ is satisfiable on some tuple from $C^{n}$, we obtain that
any $n$-type $H\in S_{n}U_{2}$ is satisfiable on some tuple from $A_{2}^{n}$%
. Using Lemma \ref{L6.3}, we obtain that $U_{2}$ is a program-saturated
structure. Evidently, $U_{1}\prec U_{3}$ and $U_{2}\prec U_{3}$. Taking into
account that $A_{1}\subseteq A_{2}\subseteq A_{3}$, it is not difficult to
show that $U_{1}$ is a substructure of the structure $U_{2}$, i.e., $U_{1} \subset U_{2}$. Using
statement (b) of Lemma \ref{L6.R.5}, we obtain $U_{1}\prec U_{2}$.
It is easy to show that $|A_2|=\beta \le \alpha (U_1)$.
Using  statement (a) of the
theorem, we obtain that the cardinality of the structure $U_{2}$ is equal to
$\alpha (U_{1})$.
\end{proof}

\section{ Computation Programs and Computation Trees \label{S6.6}}

This section is devoted to the transfer of the obtained results to programs
that are essentially close to the computation trees: if the scheme of such a program is a finite tree-scheme, then the program is an ordinary computation tree.

Let $\sigma $ be a finite or countable signature. We now define the notions
of primitive term and primitive formula of the signature $\sigma $.

A \emph{primitive term} of the signature $\sigma $ is a term of the
signature $\sigma $ of the following form:

\begin{itemize}
\item Variable  from the set $X$.

\item Constant symbol  from the signature $\sigma $.

\item Term $f(x_{i_{1}},\ldots ,x_{i_{m}})$, where $f$ is a function
symbol of arity $m$ from $\sigma $ and $x_{i_{1}},\ldots ,x_{i_{m}}\in X$.
\end{itemize}

A \emph{primitive formula} of the signature $\sigma $ is a formula of the
signature $\sigma $ of the following form:

\begin{itemize}
\item Formula $x_{i}=x_{j}$, where $=$ is the equality symbol and $%
x_{i},x_{j}\in X$.

\item Formula $p(x_{i_{1}},\ldots ,x_{i_{m}})$, where $p$ is a predicate
symbol of arity $m$ from $\sigma $ and $x_{i_{1}},\ldots ,$ $x_{i_{m}}\in X$.
\end{itemize}

\begin{definition}
Let $S=(n,G)$ be a scheme of the signature $\sigma $. This scheme will be
called a \emph{computation scheme} if each its function node is labeled with
an expression of the form $x_{i}\Leftarrow t$, where $t$ is a primitive term
of the signature $\sigma $ and $x_i \in X$, and each its predicate node is labeled with a
primitive formula of the signature $\sigma $.
\end{definition}
\begin{definition}
Let $U$ be a structure of the signature $\sigma $ and $\Gamma =(S,U)$ be a
program over the structure\emph{\ }$U$. The program $\Gamma $ will
be called a \emph{computation program} over the structure $U$ if $S$ is a
computation scheme.
\end{definition}
\begin{definition}
A computation scheme $(n,G)$ of the signature $\sigma $ will be called a
\emph{computation tree-scheme} of the signature $\sigma $ if $G$ is a tree with the root that coincides with the initial node of $G$. A computation tree-scheme $(n,G)$ of the signature $\sigma $ will be called \emph{finite} if  $G$ is a finite tree.
\end{definition}
\begin{definition}
Let $U$ be a structure of the signature $\sigma $ and $\Gamma =(S,U)$ be a
program over the structure\emph{\ }$U$. The program $\Gamma $ will
be called a \emph{computation tree} over the structure $U$ if $S$ is a finite
computation tree-scheme.
\end{definition}

\subsection{Computation-program-saturated Classes of Structures}

In this section, we consider some
results related to computation-program-saturated classes of structures.

\begin{definition}
A nonempty class $K$ of structures of the signature $\sigma $ will be called
\emph{computation-program-saturated} if any total relative to $K$
computation scheme of the signature $\sigma $ is strongly equivalent
relative to $K$ to a finite computation tree-scheme of the signature $\sigma $. 
\end{definition}

It
is clear that any program-saturated class $K$ is a
computation-program-saturated class.
Using Theorem \ref{T6.1}, we obtain the following statement.

\begin{proposition}
\label{P6.1} If a nonempty class $K$ of structures of the signature $\sigma $
has the property of compactness, then it is computation-program-saturated.
\end{proposition}

The next statement follows from Theorem \ref{T6.2}.

\begin{proposition}
\label{P6.2} Any nonempty axiomatizable class $K$ of structures of the
signature $\sigma $ is computation-program-saturated.
\end{proposition}

We already considered the following examples of axiomatizable classes of
structures:

\begin{itemize}
\item Classes of Boolean algebras, atomic Boolean algebras, and atomless
Boo\-lean algebras.

\item Classes of groups, abelian groups, abelian groups with all elements of
order $p$, where $p$ is a prime, torsion-free abelian groups.

\item Classes of commutative rings with unit, fields, fields of
characteristic $p$, where $p$ is a prime, fields of characteristic zero,
algebraically closed fields, real closed fields.
\end{itemize}

\subsection{Computation-program-saturated Structures}

In this section, we consider some results
related to computation-program-saturated structures.

\begin{definition}
Let $U$ be a structure of the signature $\sigma $. The structure $U$ will be
called \emph{computation-program-saturated} if the class $\{U\}$ is
computation-program-saturated. 
\end{definition}

The next statement follows from Theorem \ref{T6.3}.

\begin{proposition}
\label{P6.3} Let $U$ be a structure of the signature $\sigma $ and $\alpha$ be a cardinal.

{\rm (a)} If $U$ is an $\omega $-saturated structure, then $U$ is a
computation-program-saturated structure.

{\rm (b)} If $U$ is a model of the cardinality $\alpha $ of an $\alpha $%
-categorical theory, then $U$ is a computation-program-saturated
structure.
\end{proposition}

We already mentioned that each of the following structures is, for some
cardinal $\alpha $, a model of the cardinality $\alpha $ of an $\alpha $%
-categorical theory:

\begin{itemize}
\item Countable atomless Boolean algebra.

\item Abelian group with all elements of order  $p$, where $p$ is a prime number.

\item Uncountable divisible torsion-free abelian group, in particular,
additive group of real numbers $(%
\mathbb{R}
;+,0)$.

\item Uncountable algebraically closed field of the characteristic zero or $%
p $, where $p$ is a prime number, in particular, the field of complex
numbers $(%
\mathbb{C}
;+,\cdot ,0,1)$.
\end{itemize}
Each of these structures is a computation-program-saturated structure.

\subsection{Elementary Extensions}

In this section, we consider the possibility of elementary extension of the structure to a structure that is computation-program-saturated.  In Section \ref{S6.5}, for any
structure $U$, we defined a cardinal $\alpha (U)$.
The next statement follows from Theorem \ref{T6.4}.

\begin{proposition}
\label{P6.4} Let $U_{1}$ be a structure of the signature $\sigma $. Then
there exists a computation-program-saturated structure $U_{2}$, which is an
elementary extension of the structure $U_{1}$ and which cardinality is equal
to $\alpha (U_{1})$.
\end{proposition}

\section{Conclusions\label{S6}}

In this paper, we studied program-saturated classes of structures. We proved that a necessary and sufficient condition for a class to be program-saturated is its  compactness. We showed that any axiomatizable class of structures is program-saturated. We also studied individual structures, each of which forms a class that is program-saturated, and showed that such structures include all models of the cardinality $\alpha $ of $\alpha $-categorical theories.
We studied the possibility of elementary extension of the structure to a structure for which the corresponding singleton class of structures is program-saturated. We showed that this is always possible and found the minimum cardinality of such an extension. Finally, we transferred the obtained results to programs that are essentially close to computation trees.
\subsection*{Acknowledgements}

Research reported in this publication was supported by King Abdullah
University of Science and Technology (KAUST).

\bibliographystyle{spmpsci}
\bibliography{abc_bibliography}

\begin{thebibliography}{10}
\providecommand{\url}[1]{{#1}}
\providecommand{\urlprefix}{URL }
\expandafter\ifx\csname urlstyle\endcsname\relax
  \providecommand{\doi}[1]{DOI~\discretionary{}{}{}#1}\else
  \providecommand{\doi}{DOI~\discretionary{}{}{}\begingroup
  \urlstyle{rm}\Url}\fi

\bibitem{Ben-Or83}
Ben-Or, M.: Lower bounds for algebraic computation trees (preliminary report).
\newblock In: The 15th Annual {ACM} Symposium on Theory of Computing, STOC
  1983, pp. 80--86 (1983)

\bibitem{BreimanFOS84}
Breiman, L., Friedman, J.H., Olshen, R.A., Stone, C.J.: {Classification and
  Regression Trees}.
\newblock Wadsworth and Brooks (1984)

\bibitem{ChangK92}
Chang, C.C., Keisler, H.J.: Model Theory, Third Edition, \emph{Studies in Logic
  and the Foundations of Mathematics}, vol.~73.
\newblock North-Holland (1992)

\bibitem{Dobkin76}
Dobkin, D.P., Lipton, R.J.: Multidimensional searching problems.
\newblock {SIAM} J. Comput. \textbf{5}(2), 181--186 (1976)

\bibitem{Gabrielov17}
Gabrielov, A., Vorobjov, N.N.: On topological lower bounds for algebraic
  computation trees.
\newblock Found. Comput. Math. \textbf{17}(1), 61--72 (2017)

\bibitem{Grigoriev96}
Grigoriev, D., Vorobjov, N.N.: Complexity lower bounds for computation trees
  with elementary transcendental function gates.
\newblock Theor. Comput. Sci. \textbf{157}(2), 185--214 (1996)

\bibitem{Jech02}
Jech, T.: Set Theory, Third Edition.
\newblock Springer Monographs in Mathematics. Springer (2002)

\bibitem{Kleene67}
Kleene, S.C.: Mathematical Logic.
\newblock Wiley (1967)

\bibitem{Moshkov87}
Moshkov, M.: On the programs with finite development.
\newblock In: L.~Budach, R.G. Bukharajev, O.B. Lupanov (eds.) Fundamentals of
  Computation Theory, International Conference FCT'87, Kazan, USSR, June 22-26,
  1987, Proceedings, \emph{Lecture Notes in Computer Science}, vol. 278, pp.
  323--327. Springer (1987)

\bibitem{Moshkov05}
Moshkov, M.: Time complexity of decision trees.
\newblock In: J.F. Peters, A.~Skowron (eds.) Trans. Rough Sets III,
  \emph{Lecture Notes in Computer Science}, vol. 3400, pp. 244--459. Springer
  (2005)

\bibitem{Moshkov22a}
Moshkov, M.: Rough analysis of computation trees.
\newblock Discret. Appl. Math. \textbf{321}, 90--108 (2022)

\bibitem{Yao97}
Yao, A.C.: Decision tree complexity and {B}etti numbers.
\newblock J. Comput. Syst. Sci. \textbf{55}(1), 36--43 (1997)

\end{thebibliography}

\end{document}